\documentclass {llncs}
\title{Lower Bounds for Special Cases of Syntactic Multilinear ABPs}
\author{C. Ramya \and B.V.Raghavendra Rao, }
\institute{Department of Computer Science and Engineering,  IIT Madras, Chennai, India. \\  {\tt cramya2009@gmail.com, bvrr@iitm.ac.in }}

\usepackage{amssymb,amsmath}
\usepackage{xcolor,complexity}
\usepackage{enumitem}

\usepackage[location=appendix,
    appendixsectionname={Our Proofs},
    prependtoappendix=false,
    manual=true]{moveproofs}

\usepackage{listings}

\newcommand{\var}{{\sf var}}

\newcommand{\kbl}{$k_B$-hitting}
\newcommand{\kbi}{$k_{B_i'}$-hitting}

\newcommand{\rank}{{\sf rank}}

\newcommand{\pcover}{$+$-covering}
\newcommand{\scover}{signature-cover}

\newcommand{\kub}{k\mbox{-}unbalanced}
\newcommand{\kw}{k\mbox{-}weak}
\newtheorem{defn}{Definition}
\newtheorem{prop}{Proposition}
\newtheorem{obs}{Observation}

\begin{document}

\maketitle

\begin{abstract}
Algebraic Branching Programs(ABPs) are standard models for computing polynomials. Syntactic multilinear ABPs (smABPs) are restrictions of ABPs where every variable is allowed to occur at most once in every path from the start to the terminal node.
Proving lower bounds against syntactic multilinear ABPs  remains a challenging open question in Algebraic Complexity Theory.  The current best known bound is only quadratic [Alon-Kumar-Volk, ECCC 2017].

In this article we develop a new approach upper bounding the rank of the partial derivative matrix of syntactic multlinear ABPs:  Convert the ABP to a  syntactic mulilinear formula with a super polynomial blow up in the size and then exploit the structural limitations of resulting formula to obtain a rank upper bound.

 Using this approach, we prove exponential lower bounds for special cases of smABPs and circuits - namely sum of Oblivious Read-Once ABPs, $r$-pass mulitlinear ABPs and sparse ROABPs. En route, we also prove super-polynomial lower bound for a special class of syntactic multilinear  arithmetic circuits.
  
\end{abstract}

\section{Introduction}

\paragraph*{}Algebraic Complexity Theory investigates the inherent complexity of computing polynomials with   arithmetic circuit as the computational model.  Arithmetic circuits introduced by Valiant~\cite{Val79} are standard models for computing polynomials over an underlying field.
An {\em arithmetic formula} is a subclass of arithmetic circuits corresponding to arithmetic expressions. For circuits and formulas, the parameters of interest are
{\em size} and {\em depth}, where size represents the number of nodes in the graph and depth the length of longest path in the graph. The arithmetic formulas are computationally weaker than circuits, a proper separation between them is not known.

Nested in-between the computational power of formulas and circuits is yet another well-studied model for computing polynomials referred to as {\em Algebraic Branching Programs} (ABPs for short). 
We know, 

\begin{center}
Arithmetic Formula $\subseteq_{\P}$ ABP $\subseteq_{\P}$ Arithmetic Circuits.
\end{center}

where the subscript $\P$ denotes the containment upto polynomial blow-up in size. Most of algebraic complexity theory revolves around understanding whether these containments are strict or not.


Separation of complexity classes of polynomials involves obtaining lower bound for specific polynomial against classes of arithmetic circuits. 
For general classes of arithmetic circuits, Baur and Strassen~\cite{BS83} proved that any arithmetic circuit compuitng an explicit $n$-variate degree $d$ polynomial must have size $\Omega(n\log d)$. In fact, this is the only super linear lower bound we know for general arithmetic circuits. 

While the challenge of proving lower bounds for general classes of circuits still seems to be afar, recent research has focused on circuits with additional structural restrictions such as multilinearity, bounded read etc. We now look at some of the models based on these restrictions in more detail.


 An arithmetic  circuit~(formula,ABP) is said to be {\em multilinear}  if every   gate (node) computes a multilinear polynomial. A seminal work of Raz~\cite{Raz09} showed that  multilinear formulas computing $\det_n$ or $perm_n$ must have size $n^{\Omega(\log n)}$. 
Although we know strong lower bounds for multilinear formulas,  the best known lower bound against syntactic multilinear circuits  is almost quadratic in the number of variables~\cite{AKV17}. 
 Note that any   multilinear ABP of $n^{O(1)}$ size computing $f$ on $n$ variables can be converted to a multilinear formula of size $n^{O(\log n)}$ computing $f$. In order to prove super-polynomial lower bounds for ABPs, it is enough to obtain  a multilinear formula computing $f$ of size $n^{o(\log n)}$ or  prove a lower bound of $n^{\omega(\log n)}$ for multilinear formulas, both of which are not known.

Special cases of multilinear ABPs have been studied time and again. In this work, we focus on the class of Read-Once Oblivious Algebraic branching programs~(ROABP for short). ROABPs are ABPs where every edge is labeled by a variable and every variable appears as edge labels in atmost one layer.  There are explicit polynomials with  $2^{\Omega(n)}$ ROABP size lower bound  \cite{Nis91,Jan,KNS16}. Also, ROABPs have been well studied in the   context of polynomial identity testing algorithms ~(See e.g.,\cite{For14}) 

In this article, we prove lower bounds against sum of multilinear ROABPs and other classes of restricted multilinear ABPs and circuits. Definitions of the models considered in this article can be found in Section~\ref{sec:prelim}.

\paragraph*{Our Results}
Let $X=\{x_1,\ldots,x_N\}$ and $\mathbb{F}$ be a field.  Let $g$ denote the family of $N$ variate (for $N$ even) defined by Raz and Yehudayoff~\cite{RY08}. (See Definition~\ref{def:raz-poly} for more details.)
As our main result, we show that any sum of sub-exponential ($2^{o(N^{\epsilon})}$)  size  ROABPs  to represent $g$ requires  $2^{N^{\epsilon}}$ many summands: 

\begin{theorem}
\label{thm:lb-roabp}
Let $f_1,\ldots f_m$ be polynomials computed by oblivious ROABPs such that $g= f_1+\cdots + f_m$. Then, $m = \frac{2^{\Omega(N^{1/5})}}{s^{c\log N}}$, where $c$ is a constant and $s=\max\{s_1,s_2,\ldots,s_m\}$, $s_i$ is the size of the ROABP computing $f_i$. 
\end{theorem}

 Further, we show that Theorem~\ref{thm:lb-roabp} extends to the case of $r$-pass multilinear ABPs (Theorem~\ref{thm:lb-rpass}) for $r=o(\log n)$ and $\alpha$-sparse  multilinear ABPs (Theorem~\ref{thm:sparse-factor}) for $1/1000 \leq \alpha \leq 1/2$.

Finally, we develop a refined approach to analyze syntactic  multilinear formulas based on the central paths introduced by Raz~\cite{Raz09}.  Using this, we prove exponential lower bound against a class of $O(\log N)$  depth syntactic multilinear circuits (exact definition can be found in Section~\ref{sec:signature}, Definition~\ref{def:variable-close}). 

\begin{theorem}
\label{thm:lb-delta-close}
Let $\delta < N^{1/5}/10$ and $c= N^{o(1)}$. Any $O(\log N)$ depth $(c,\delta)$ variable close syntactically multilinear  circuit computing the polynomial $g $ requires  size $ 2^{\Omega(N^{1/5}/\log N)}$.
\end{theorem}
\paragraph*{Our approach}
Our proofs are a careful adaptation of the rank argument developed by Raz~\cite{Raz09}. This involves upper bounding the dimension of the partial derivative matrix (Definition~\ref{def:partition}) of the given model under a random partition of variables. However, upper bounding the rank of the partial derivative matrix of a syntactic multilinear ABP is a difficult task and there are no known methods for the same. To the best of our knowledge, there is no non-trivial upper bound on the rank of the partial derivative matrix of polynomials computed by ABPs (or special classes of ABPs) under a random partition. 

Our crucial observation is, even though conversion of a syntactic  multilinear ABP  of size $s$ into a syntactic multilinear formula  blows the size to  $s^{O(\log s)}$, the resulting formula is much simpler in structure than an arbitrary syntactic  multilinear formula of size $n^{O(\log s)}$.  For  each of the special classes of multilinear ABPs (ROABPS, $r$-pass ABPs etc) ) considered in the article, we identify and exploit the  structural limitations of the formula obtained from the corresponding  ABP  to prove upper bound on the rank of the partial derivative matrix under a random partition. Overall our approach to upper bound the rank can be summarized as follows:
\begin{enumerate}
\item Convert the given multilinear ABP $P$  of size $s$ to a multilinear formula $\Phi$ of size $s^{O(\log s)}$ (Lemmas~\ref{lem:abptoformula},~\ref{lem:rpasstoformula} and ~\ref{lem:sparsetoformula});
\item Identify structural limitations of the resulting formula $\Phi$ and exploit it to prove upper bound on the rank of the partial derivative matrix under a random partition (Lemmas~\ref{lem:kbl},~\ref{lem:rankub}, ~\ref{lem:sparse-ub} and ~\ref{lem:covering-sign});
\item Exhibit a hard polynomial that has full rank under all partitions. (Lemma~\ref{lem:ry}.)
\end{enumerate}

\subsubsection*{Related Results} 
 Anderson et. al~\cite{AFSSV16} obtained exponential lower bound against oblivious read $k$ branching programs.   Kayal et. al~\cite{KNS16}  obtained a polynomial that can be written as sum of three ROABPs each of polynomial size such that any ROABP computing it has exponential size. Arvind and Raja~\cite{AR16} show that if permanent can be written as a sum of $N^{1-\epsilon}$ many ROABPs, then at least one of the ROABP must be of exponential size.   Further,  sum of read-once polynomials, a special class of oblivious ROABPs was considered by   Mahajan and Tawari~\cite{MT15}, independently by the authors~\cite{CR16}. 
Recently, Chillara et. al~\cite{CLS18} show that any $o(\log N)$ depth syntactic multilinear circuit cannot  a  polynomial that is computable by  width-2 ROABPs.

The existing lower bounds against ROABPs  or sm-ABPs,
 implicitly restrict the  number  of  different orders in which the variables can  be read  along any $s$ to $t$ path.  In fact, the lower bound given in Arvind and Raja~\cite{AR16} allows only  $N^{1-\epsilon}$ different ordering of the variables.  To the best of our knowledge, this is the state of art with respect to the number of  variable orders allowed in ABPs. Without any restriction on the orderings, the best known lower bound is only   quadratic upto poly logarithmic factors~\cite{AKV17}. In this light, our results in Theorems~\ref{thm:lb-roabp} and~\ref{thm:lb-rpass} can be seen as the first of the kind where the number of different orders allowed is sub-exponential.


Proofs omitted due to space constraints can be found in the Appendix.

\section{Preliminaries}
\label{sec:prelim}

In this section we include necessary definitions and notations used. We begin with the formal definition of the models considered in this article.

An {\em arithmetic circuit}~$\mathcal{C}$ over a field  $\mathbb{F}$ and variables  $X={x_1,\ldots, x_N}$ is a directed acyclic graph with vertices of in-degree 0 or 2 and exactly one vertex of out-degree 0 called the output gate. The vertices of in-degree 0 are called input gates and are labeled  by elements from $X \cup \mathbb{F}$. The vertices of in-degree 2 are labeled by either $+$ or $\times$. Every gate in $\cal{C}$ naturally computes a polynomial. The polynomial $f$ computed by $\mathcal{C}$ is the polynomial computed by the output gate of the circuit. The {\em size} of an arithmetic circuit is the number of gates in $\mathcal{C}$ and {\em depth} of $\mathcal{C}$ is the length of the longest path from an input gate to the output gate in $\mathcal{C}$. 
An {\em arithmetic formula} is an arithmetic circuit where the underlying undirected graph is a tree.

 An {\em Algebraic Branching Program} $P$ (ABP for short) is a layered  directed acyclic graph with two special nodes, a start node $s$ and a terminal node $t$.   
  Each edge in $P$ is labeled by either an $x_i\in X$ or $\alpha\in\mathbb{F}$.  The size of $p$ is the total number of nodes, width is the maximum number of nodes in any layer of $P$.  Each path $\gamma$ from $s$ to $t$ in $P$ computes the product of the labels of the edges in $\gamma$ which is a polynomial. The ABP $P$ computes the sum over all $s$ to $t$ paths of such polynomials. 

An ABP $P$ is said to be {\em syntactic multilinear} (sm-ABP for short)  if every  variable occurs at most once in every   path in  $P$. 
An ABP is said to be  {\em oblivious} if for every layer $L$ in $P$ there is at most one variable that labels edges from $L$.

\begin{defn}{\em (Read-Once Oblivious ABP.)}
An ABP $P$ is said to be Read-Once Oblivious (ROABP for short) if $P$ is an oblivious and each $x_i\in X$ appears as edge label in at most one layer. 
\end{defn}

In any Oblivious ROABP, every variable appears in exactly one layer and all variables in a particular layer are the same. Hence, variables appear in layers from the start node to the terminal node in the
{\em variable order} $x_{i_1},x_{i_2},\ldots,x_{i_n}$ where $(i_1,i_2,\ldots,i_n)\in S_n$ is a permutation on $[n]$. 
 A natural generalization of ROABPs is the $r$-pass ABPs defined in~\cite{AFSSV16}:

\begin{defn}{\em ($r$-pass multilinear ABP).}
An oblivious  sm-ABP $P$ is said to be $r$-pass  if there   are permutations $\pi_1,\pi_2,\ldots,\pi_r\in S_n$ such that $P$ reads the variables from $s$ to $t$ in the order
$(x_{\pi_1(1)},x_{\pi_1(2)},\ldots,x_{\pi_1(n)}),\ldots$,$(x_{\pi_r(1)},x_{\pi_r(2)},\ldots,x_{\pi_r(n)}).$
\end{defn}

Recall that a  polynomial $f\in\mathbb{F}[X]$ is $s$-sparse if it has at most $s$ monomials with non-zero coefficients.
\begin{defn}{\em ($\alpha$-Sparse   ROABP).}~\cite{For14}
An $d+1$ layer ABP $P$ is said to be an $\alpha$-sparse ROABP if there is a partition of $X$ into $d = \Theta(N^{\alpha})$ sets $X_1,X_2,\ldots,X_d$  with $|X_i| = N/d$ such that every edge label in layer $L_i$ is an $s$-sparse multilinear polynomial in $\mathbb{F}[X_i]$ for  $s=N^{O(1)}$.
\end{defn}


Let $\Psi$ be a circuit over $\mathbb{F}$ with $X=\{x_1,\ldots, x_N\}$ as inputs.  For a gate $v$ in $\Psi$, let $X_{v}$ denote the set of variables that appear  in the sub-circuit rooted at $v$. The circuit $\Psi$ is said to be {\em syntactic multilinear} (sm for short), if for every $\times$ gate $v= v_1\times v_2$ in $\Psi$, we have $X_{v_1} \cap X_{v_2} =\emptyset $. By definition, every syntactic multilinear circuit is a multilinear circuit. In~\cite{Raz09}, it was shown that every  multilinear formula can be transformed into a syntactic multilinear formula of the same size, computing the same polynomial.

Let $\Psi$ be a circuit (formula) and $v$ be a gate in $\Psi$. The {\em product-height} of $v$ is the maximum number of $\times$  gates   along any $v$ to root path in $\Psi$.

We now review the partial derivative matrix of a polynomial introduced in~\cite{Raz09}. Let  $Y = \{ y_1,\ldots,y_m\}$ and $Z=\{z_1,\ldots,z_m\}$ be  disjoint sets of variables. 

\begin{defn}{\em (Partial Derivative Matrix.)}
\label{def:partition}
Let $f\in\mathbb{F}[Y,Z]$ be a polynomial. The {\em partial derivative matrix} of $f$(denoted by $M_f$) is a $2^m\times 2^m$ matrix defined as follows. For monic multilinear monomials $p$ and $q$ in variables $Y$ and $Z$ respectively, the entry $M_f [p,q]$ is the coefficient of the monomial $pq$ in $f$.  
\end{defn}

For a polynomial $f$, let  $\rank(M_f)$ denote the rank of the matrix $M_f$ over the field $\mathbb{F}$. It is known that  $\rank(M_f)$ satisfies sub-additivity and sub-multiplicativity:
 
\begin{lemma}{{\em \cite{Raz09}(Sub-additivity, sub-multiplicativity)  .}}
\label{lem:sub-aditivity}
Let $f,g \in \mathbb{F}[Y,Z]$. Then, we have that
$\rank(M_{f+g}) \leq \rank(M_f)+\rank(M_g).$ Further, if $\var(f) \cap \var(g) = \emptyset$, then $\rank(M_{fg}) = \rank(M_f)\rank(M_g)$.
 \end{lemma}
Further,  since row-rank of a matrix is equal to its column rank, we have:

\begin{lemma}{\em \cite{Raz09}} 
\label{lem:rankub}
 For  $f\in\mathbb{F}[Y_1,Z_1]$,  $\rank(M_f)\leq 2^{\min\{|Y_1|,|Z_1|\}}$, where $Y_1\subseteq Y, Z_1\subseteq Z$. 
\end{lemma}

For $f\in\mathbb{F}[X]$, it may be noted that the parital derivative matrix $M_f$  is dependent on the partition of the variable set $X$ into variables in $Y\cup Z$. In most of the cases, partition of the variable set is not apparent. In such cases, we need to consider a distribution over the set of all such partitions. We represent a partition as a bijective function $\varphi : X \rightarrow Y \cup Z$, where $|Y| = |Z| = |X|/2$. 

Let ${\cal D}$ be the  uniform distribution on the set of  all partitions $\varphi: X \to Y \cup Z$, with $|Y| = |Z| = |X|/2$.

Now, we state a useful property of the standard hypergeometric distribution that  will be needed  later.

\begin{prop}
\label{prop:hyper}{\em \cite{Raz06,RP15}}
{\em(Hypergeometric Distribution).} Let $M_1,M_2\leq S$ be integers. Let ${\cal{H}}(M_1,M_2,S)$ denote the distribution of size of the intersection of a random set of size $M_2$ and a set of size $M_1$ in a universe of size $S$. Let $\chi$ be a random variable distributed according to ${\cal{H}}(M_1,M_2,S)$ :
\begin{enumerate}
\item If $S^{1/2} \leq M_1 \leq S/2$ and $S/4 \leq M_2 \leq 3S/4$ then $\Pr[\chi=a]\leq O(S^{-1/4})$.
\item If $0 \leq M_1 \leq 2S/3$ and $S/4 \leq M_2 \leq 3S/4$ then $\Pr[\chi=a]\leq O(M_1^{-1/2})$ for any $a\leq M_1$.
\end{enumerate}
\end{prop}

We consider the full rank polynomial $g$ defined  by Raz and Yehudayoff~\cite{RY08}  to prove lower bounds for all models that arise in this work. 

\begin{defn}{\em (Hard Polynomial.)}
\label{def:raz-poly}
Let $N\in\mathbb{N}$ be an integer. Let $X=\{x_1,\ldots, x_{N}\}$ and $\mathcal{W} = \{w_{i,k,j} \}_{i,k,j\in[N]}$. For any two integers $i,j\in\mathbb{N}$, we define an interval $[i,j] = \{ k\in\mathbb{N}, i\leq k\leq j \}$. Let $|[i,j]|$ be the length of the interval $[i,j]$. Let $X_{i,j} = \{ x_p \mid p\in [i,j]\} $ and $W_{i,j}=\{ w_{i',k,j'}\mid i',k,j'\in[i,j] \}$. Let  $\mathbb{G}=\mathbb{F}(\mathcal{W})$, the rational function field. For every $[i,j]$ such that $|[i,j]|$ is even we define a polynomial $g_{i,j}\in\mathbb{G}[X]$  as 
 $g_{i,j}=1$ when   $|[i,j]|=0$  and 
 if $|[i,j]|>0$ then, {\small $g_{i,j }\triangleq (1+x_ix_j)g_{i+1,j-1} + \sum_{k}w_{i,k,j}g_{i,k}g_{k+1,j}.$}
where $x_k$, $w_{i,k,j}$ are distinct variables, $1\le k\le j$ and the summation is over  $k\in [i+1,j-2]$ such that  $|[i,k]|$ is  even.  Let $g\triangleq g_{1,N}$.
\end{defn}

\begin{lemma}\cite[Lemma~4.3]{RY08}
\label{lem:ry} Let $X=\{x_1,\ldots, x_{N}\}$ and $\mathcal{W} = \{w_{i,k,j} \}_{i,k,j\in[N]}$.  Let $\mathbb{G}=\mathbb{F}(\mathcal{W})$ be the set of rational functions over field $\mathbb{F}$ and $\mathcal{W}$. Let $g\in\mathbb{G}[X]$ be the polynomial in Definition \ref{def:raz-poly}. Then for any $\varphi\sim {\cal D}$, 
$\rank(M_{g^\varphi})= 2^{N/2}$.
\end{lemma}


\section{Lower Bounds for Special cases of  sm-ABPs}

\label{sec:explb}
In this section, we obtain exponential lower bound for sum of ROABPs and related special classes of syntactic multilinear ABPs. 
\subsection{Sum of  ROABPs: Proof of Theorem~\ref{thm:lb-roabp}}
\label{subsec:sumroabp}
Let $P$ be an ROABP with $\ell+1$ layers $L_0,L_1,L_2,\ldots,L_{\ell}$ computing a multilinear polynomial $f\in\mathbb{F}[x_1,x_2,\ldots,x_N]$. For every $i\in \{0,1,\ldots,\ell-1\}$, we say a layer $L_i$ is a {\em  constant} layer if every edge  going out of a  vertex in $L_i$ is labeled by a constant from $\mathbb{F}$, else we call the layer $L_i$ a {\em variable} layer. For any {\em variable} layer $L_i$ denote by $\var(L_i)$ the variable in $X$ that labels edges  going out of  vertices in $L_i$. For nodes $u,v$ in $P$, we denote by $[u,v]$ the polynomial computed by the subprogram with $u$ as the start node  and  $v$ as the terminal node and let  $X_{u,v}$  be the set of variables that occur in $P$ between layers containing $u$ and $v$ respectively. We can assume without loss of generality that $P$ does not have any two consecutive constant layers and that every ROABP $P$ has exactly $2N$ layers by introducing dummy constant layers in between consecutive variable layers. Further, we assume that the variables occur in $P$ in the order $x_1,\ldots x_N$, and hence  indices of variables in $X_{u,v}$  is an interval   $[i,j]=\{t\in\mathbb{N}\mid i\leq t\leq j\}$ for some $i< j$. (In case of a different order $\pi$ for occurrence of variables,                                                                                                  
the interval would be $[i, j ] = \{\pi(i), \pi(i+1), \ldots, \pi(j)\}$.)

\noindent{\bf Approach:}  In order to prove Theorem~\ref{thm:lb-roabp}, we use $\rank(M_{f^\varphi})$ as a complexity measure, where  $\varphi\sim {\cal D}$. The outline is as follows:

\begin{enumerate}
\item Convert the ROABP $P$ into a multilinear formula $\Phi$ with a small (super polynomial) blow up in size (Lemma~\ref{lem:abptoformula}).
\item  Obtain a partition   $B_1,\ldots, B_t$  of the variable set with $O(\sqrt{N})$ parts of almost equal size, so that there is at least one set that is highly unbalanced under a random $\varphi$ drawn from ${\cal D}$. (Observation~\ref{obs:blocks} and Lemma~\ref{lem:kbl}.)
\item Using the structure of the formula $\Phi$, show that if at least on of the $B_i$ is highly unbalanced, then the formula $\Phi$ has low rank (Lemma~\ref{lem:rank}).
\item Combining with Lemma~\ref{lem:ry} gives the required lower bound. 
\end{enumerate}
The following lemma lists useful properties of the straightforward conversion of an  ROABP into a multilinear formula: 

\begin{lemma}
\label{lem:abptoformula}
Let $P$ be an    ROABP of size $s$  computing a polynomial $f\in\mathbb{F}[x_1,\ldots,x_N]$. Then $f$ can be computed by a syntactic multilinear formula $\Phi$  of size $s^{O(\log N)}$ and depth $O(\log N)$ such that 
\begin{enumerate}
\item $\Phi$ has an alternative of layers of $+$ and $\times$ gates; and
\item $\times$ gates have  fan-in bounded by two; and
\item Every  $+$ gate $g$ in $\Phi$ computes a polynomial $[u,v]$ for some  $u, v$ in $P$; and
\item Every $\times$ gate computes a product $[u,v]\times [v,w]$, for some  $u,v$ and $w$ in $P$.
\item The root of $\Phi$ is a $+$ gate. 
\end{enumerate}  
\end{lemma}

\makeproof{lem:abptoformula}{
The proof is a simple divide and conquer conversion of branching programs to formulas.
Let $P$ be an   ROABP with $\ell+1$ layers $L_0,L_1,\ldots,L_{\ell}$ with   $s$  and  $t$ as the start and terminal nodes respectively. Let $L_i$ be such  that $|\var(L_0)\cup\var(L_1) \cup  \cdots \cup \var(L_i)|, |\var(L_{i+1})\cup \cdots\cup \var(L_\ell)| \in \{\lceil{N/2}\rceil, \lfloor N/2\rfloor\}$ and  $u_{i_1},u_{i_2},\ldots,u_{i_k} (k\leq s)$ be the nodes at the layer $L_i$. Then,
\begin{equation}
\label{eqn:formula1}
f = \sum\limits_{j=1}^{k} [s,u_{i_j}]\times [u_{i_j},t]
\end{equation}
where $[u,v]$ is the polynomial computed by the subprogram with start node $u$ and $v$ as the terminal node. By induction on $N$,  Let $\phi_j$ (respectively $\psi_j$) be the formula computing $[s,u_{i_j}]$ (respectively $[u_{i_j},t]$ ).
Then $\Phi = \sum_{j=1}^k \phi_j\times \psi_j.$ By induction, it follows  that the resulting formula $\Phi$ has size $s^{O(\log N)}$, depth $O(\log N)$ and is syntactic multilinear. Also,  by the construction above, it can be verified that $\Phi$ satisifes the conditions $1$ to $5$. \qed
}

Let $P$ be an   ROABP  and $\Phi$ be the syntactic multilinear formula obtained from $P$ as in Lemma~\ref{lem:abptoformula}. Let $g$ be a  $+$ (respectively $\times$) gate in $\Phi$ computing $[u_g, v_g]$ (respectively $[u_g,v_g]\times [v_g, w_g]$) for some nodes $u_g$, $v_g$ and $w_g$ in $P$. Since $P$ is an ROABP with variable order $x_1,x_2,\ldots x_N$, the set $X_{u_g, v_g}$ (respectively $ X_{u_g,v_g} \cup X_{v_g, w_g}$) corresponds to an interval $I_g$ in $\{1,\ldots, N\}$. We call $I_g$ the {\em interval associated with } $g$. By the construction of $\Phi$ in  Lemma \ref{lem:abptoformula}, the intervals  have the following properties : 
\begin{enumerate}
\item For any gate $g$ in $\Phi$ at product-height $i$, $|I_g| \in [ N/2^i - i,N/2^i+i]$. 
\item For any $+$ gate $g$ in $\Phi$ with children $g_1,\ldots,g_w$, we have $I_g=I_{g_1}=\cdots=I_{g_w}$. 
\item Let $\cal{I}$ be the set of all distinct intervals associated with gates at product-height $\frac{\log N}{2}$ in $\Phi$. The intervals in ${\cal I}$ are disjoint and $|{\cal{I}}|= \Theta(\sqrt{N})$. For any $I_j\in{\cal{I}}$, $\sqrt{N} -\log N \leq |I_j|\leq \sqrt{N}+\log N$.
\end{enumerate}

We call the intervals in ${\cal{I}}$ as {\em blocks} $B_1,B_2,\ldots,B_{t}$ in $\Phi$ where $t=\Theta(\sqrt{N})$.  For any block $B_\ell=[i_\ell,j_\ell]$, $X_{\ell} = \{x_{i_a}\mid i_\ell\leq i_a\leq j_\ell\}=\var(L_{i_\ell})\cup \var(L_{i_\ell+1})\cup \cdots \cup \var(L_{j_\ell})$.

Let $\varphi: X \to Y \cup Z$ be a partition. We say a block $B_\ell$ is $\kub$ with respect to $\varphi$ iff $||Y\cap \varphi(X_\ell)|-|Z\cap \varphi(X_\ell)||>k$. For any two intervals $I_1=[i_1,j_1]$ and $I_2=[i_2,j_2]$ we say $I_1\subseteq I_2$ iff $i_2\leq i_1\leq j_1\leq j_2$.   

\begin{obs}
\label{obs:blocks}
Let $P$ be an ROABP and $\Phi$ be the syntactic multilinear formula obtained from $P$  and $B_1,\ldots, B_t$ be the blocks in $\Phi$. Then, for any gate $v$ in $\Phi$, 
\begin{itemize}
\item[(1)] If $v$ is at a product-height $<\frac{\log N}{2}$ in $\Phi$, then $B_i\subseteq I_v$ for some block $B_i$.
 \item[(2)]If $v$ is at product-height $>\frac{\log N}{2}$ in $\Phi$,  then for every $1 \le i \le t$, either $I_v \subseteq B_i$ or $B_i \cap I_v = \emptyset$.
 \item[(3)]If $v$ is at product-height $\frac{\log N}{2}$ in $\Phi$,  then for every $1 \le i \le t$, either $I_v = B_i$ or $B_i \cap I_v = \emptyset$. 
\end{itemize}
\end{obs}
We need the following before formalizing Step 3 in the approach outlined. 
\begin{defn}({\kbl~formula}.)
Let $\varphi: X \to Y \cup Z$  be a partition and $B$ be a $k$-unbalanced block in $\Phi$ with respect to $\varphi$.  A gate $v$ with product-height $\le \frac{\log N}{2}$  in  $\Phi$ is \kbl~  if either 
\begin{itemize}
\item[(i)] $I_v=B$; Or 
\item[(ii)] $B\subseteq I_v$ and, 
 \begin{itemize}
  \item  If  $v$ is a sum gate with children $v_1,\ldots,v_w$, the  gates $v_1,\ldots,v_w$ are \kbl. 
  \item If $v$ is a product gate with children $v_1,v_2$, then atleast one of $v_1$ or $v_2$ are \kbl.    
\end{itemize}
\end{itemize}
A formula $\Phi$ is \kbl\ with respect to  $\varphi$~ if the root   $r$   is \kbl\ for some $k$-unbalanced block $B\in \{ B_1,B_2,\ldots,B_t\}$ where $t=\Theta(\sqrt{N})$.
\end{defn}

In the following, we note that the partial derivative matrix of  \kbl~ formulas have low rank:

\begin{lemma}
\label{lem:rank}
Let $P$ be an ROABP computing $f$ and $\Phi_P$ be the multilinear formula obtained from $P$ computing $f$. Let $\varphi\sim{\cal D}$ such that block $B$ is \kub~ in $\Phi$ with respect to $\varphi$. Let $v$ be a gate in $\Phi$ that is \kbl~then $\rank(M_{f_v^\varphi})\leq |\Phi_v|\cdot 2^{|X_v|/2-k/2}$.
\end{lemma}
\makeproof{lem:rank}{
Proof is by induction on the structure of the formula. \\
For the base case, let $v$ be a gate in $\Phi$ at product-height  $(\log N)/2$.  By Observation \ref{obs:blocks}, either $I_v =  B$ or $I_v\cap B=\emptyset$. As $v$ is \kbl, $I_v = B$. Since $B$ is $\kub$, we have $X_{v}$ is $\kub$. By Lemma \ref{lem:rankub}, $\rank(M_{f_v^\varphi})\leq 2^{\min\{|Y_v|,|Z_v|\}} \leq 2^{|X_v|/2-k/2}$. For the induction step, let $v$ be a node at product depth $ \geq (\log N)/2$. 
\begin{description}
\item[Case~1]$v$ is a product gate with two children $v_1,v_2$. Since $v$ is \kbl, atleast one of $v_1$ or $v_2$ is \kbl. Without loss of generality let $v_1$ be \kbl. By induction hypothesis, $\rank(M_{f_{v_1}^\varphi})\leq |\Phi_{v_1}|\cdot2^{|X_{v_1}|/2-k/2}$ and $\rank(M_{f_{v_2}^\varphi})\leq 2^{|X_{v_2}|/2}$. Then $\rank(M_{v})\leq \rank(M_{v_1})\cdot \rank(M_{v_2}) \leq |\Phi_{v_1}|\cdot2^{|X_{v_1}|/2+|X_{v_2}|/2-k/2}\leq |\Phi_v|\cdot 2^{|X_v|/2-k/2}$ as $X_v=X_{v_1}\cup X_{v_2}$.
\item[Case~2] $v$ is a sum gate with children $v_1,v_2,\ldots,v_w$. Since $v$ is \kbl, every child of $v$ is \kbl. Then by induction hypothesis, $\rank(M_{v_i})\leq |\Phi_{v_i}|\cdot2^{|X_{v_i}|/2-k/2}$. As $X_{v_1}=X_{v_2}=\cdots=X_{v_w}$, $\rank(M_{v})\leq |\Phi_v|\cdot 2^{|X_v|/2-k/2}$. \qed
\end{description}
}

\begin{obs}
\label{obs:not-kbl}
Let $\varphi: X \to Y \cup Z$  be a partition and $B$ be a $k$-unbalanced block in $\Phi$ with respect to $\varphi$. 
\begin{enumerate}
\item If a $+$ gate $v$ in $\Phi$ with children $v_1,\ldots,v_w$ is not \kbl~  then $I_{v_j}\cap B =\emptyset$ for some $j\in[w]$. 
\item If a $\times$ gate $v$  with children $v_1,v_2$ is not \kbl~  then $I_{v_1}\cap B =\emptyset$ and $I_{v_2}\cap B =\emptyset$.
\end{enumerate}
\end{obs}

Further, we observe that, proving that a formula $\Phi$ is \kbl\ with respect to a partition, is equivalent to showing existence of a $k$-unbalanced block among $B_1,\ldots, B_t$.
\begin{obs}
\label{obs:hitting}
Let $B_1,\ldots, B_t$ be the blocks of the formula $\Phi$ obtained from an   ROABP $P$. Let $B\in\{B_1,\ldots, B_t\}$ be a $k$-unbalanced block with respect to a partition $\varphi$. Then, $\Phi$ is \kbl\ with respect to $\varphi$.
\end{obs}
\makeproof{obs:hitting}{
Suppose not, $B\in\{B_1,\ldots, B_t\}$ be a $k$-unbalanced block with respect to a partition $\varphi$ and $\Phi$ is not \kbl\ with respect to $\varphi$. Let gate $g$ be at product-height $(\log N)/2$ in $\Phi$ such that $I_g=B$. Since $\Phi$ is not \kbl\, root gate $r$ of $\Phi$ is not \kbl. We know $r$ is a $+$ gate with children say $r_1,r_2,\ldots,r_w$. By Observation \ref{obs:not-kbl}, there exists $i\in[w]$ such that $r_i$ is not \kbl\ i.e. $I_{r_i}\cap B=\emptyset$. Also as $r$ is a $+$ gate, $I_{r_1}=I_{r_2} = \cdots = I_{r_w}$. This implies that none of $r_1,r_2,\ldots,r_w$ are \kbl. $r_1,r_2,\ldots,r_w$ being product gates,  $r_1,r_2,\ldots,r_w$ are not \kbl\ implies that none of their children are \kbl\ by Observation \ref{obs:not-kbl}. In this way, we get that no descendant of $r$ is \kbl\ which is a contradiction to the fact that gate $g$ is \kbl. \qed
}
In the remainder of the section, we estimate the probability  that at least  one of the blocks among  $B_1,\ldots, B_t$ is $k$-unbalanced. 

\begin{lemma}
\label{lem:kbl}
Let $P$ be an ROABP computing a polynomial $f\in\mathbb{F}[x_1,\ldots,x_N]$ and $\Phi_P$ be the syntactic multilinear formula computing $f$. Let $\varphi\sim{\cal D}$. Then, for any $k \le N^{1/5}$, there exists a block $B$ in $\Phi$ such that such that
$$\Pr\limits_{\varphi\sim {\cal D}}[\text{$\Phi$ is \kbl~}]\geq 1-2^{-\Omega(\sqrt{N}\log N)}$$
\end{lemma}

\makeproof{lem:kbl}{
By Observation \ref{obs:hitting}, $\Pr\limits_{\varphi\sim {\cal D}}[\text{$\Phi$ is \kbi~}]\geq \Pr[\exists~i,\text{$B_i$ is $\kub$}]$. Here we1 estimate $\Pr[\exists~i,\text{$B_i$ is $\kub$}]$. Let $P$ be an ROABP and $B_1,\ldots,B_{t}$ be blocks in $\Phi$.  Note that for any $\ell\in[t],~\sqrt{N} - \log N\leq|X_\ell|\leq \sqrt{N} + \log N$. Let ${\cal{E}}_i$ be the event that the block $B_i$ is not \kub.  For any block $\ell\in[t]$, denote $Y_\ell =\varphi(X_\ell)\cap Y$. Let $\chi =|Y_\ell|$ be a  random variable.  Observe that $\chi$ has the distribution ${\cal{H}}(S,M_1,M_2)$ with \begin{align*}
S &= N-(|X_1|+\cdots+|X_{\ell-1}|) \in [ N-(\ell-1)(\sqrt{N} + \log N),  N-(\ell-1)(\sqrt{N} - \log N)] \\
M_1 &= |X_\ell| \\
M_2 &= N/2-(|Y_1|+\cdots+|Y_{\ell-1}|)\in [ N/2-(\ell-1)(\sqrt{N} + \log N),  N/2-(\ell-1)(\sqrt{N} - \log N)]
\end{align*}
For $(\ell-1)<\sqrt{N}/4$, we have :
\begin{itemize}
\item[(i)] $3N/4\leq S \leq N$; and 
\item[(ii)] $S/4\leq N/4\leq M_2 \leq N/2\leq 2S/3 \leq 3S/4 $; and
\item[(iii)] have $S^{1/2}\leq \sqrt{N}\leq M_1 \leq 2\sqrt{N}\leq 3N/8 \leq S/2$ for large enough $N$.
\end{itemize}
By Proposition \ref{prop:hyper} (1), we have
$\Pr[\chi=a]\leq O(S^{-1/4}) = O(N^{-1/4})$. Therefore, for $i<\sqrt{N}/4,~\Pr[{\cal{E}}_i] \leq O(k\cdot N^{-1/4}) = O(N^{-1/20})$ for $k\leq N^{1/5}$. 
Let ${\cal{E}}$ be the event that for all $i\in[\sqrt{N}/4]$, block $B_i$ is not \kub. 
\begin{align*}
\mathcal{E} &= \mathcal{E}_1 \cap \mathcal{E}_2 \cap \cdots \cap \mathcal{E}_{\sqrt{N}/4} \\
\Pr[\mathcal{E}] &= \Pr[\mathcal{E}_1 \cap \mathcal{E}_2 \cap \cdots \cap \mathcal{E}_{\sqrt{N}/4}] \\
&= \Pr[\mathcal{E}_1]\cdot \prod\limits_{i=2}^{\sqrt{N}/4}\Pr[\mathcal{E}_i\mid \cap_{j=1}^{i-1}\mathcal{E}_{j}] \\
& \leq O(2^{-\sqrt{N}\log N/80})
\end{align*}
Note $\bar{{\cal{E}}}$ is the event that there exists an $i\in \sqrt{N}/4$ such that $B_i$ is \kub. $\Pr[\bar{{\cal{E}}}] =1- \Pr[{\cal{E}}]\geq 1- \frac{1}{2^{-\frac{1}{80}\sqrt{N}\log N}}$. \qed
}
 
\begin{corollary}
\label{cor:rank1}
Let $P$ be an ROABP and $\Phi_P$ be the multilinear formula obtained from $P$ computing $f$. Let $\varphi\sim{\cal D}$. Then with probability  $1-2^{-\Omega(\sqrt{N}\log N)}$, $\rank(M_{f^\varphi})\leq |\Phi|\cdot 2^{N/2-N^{1/5}}$.
\end{corollary}
\begin{proof}
Follows directly from Lemmas \ref{lem:rank} and \ref{lem:kbl}. \qed
\end{proof}
 
We are ready to combine the above to prove Theorem~\ref{thm:lb-roabp}:
\begin{proof}[of Theorem~\ref{thm:lb-roabp}]
Suppose, $f_i$ has an   ROABP $P_i$ of size $s_i$. Then, by Lemma~\ref{lem:abptoformula}, there is a multilinear formula $\Phi_i$ computing $f_i$. By Lemma~\ref{lem:kbl},  probability that $\Phi_i$ is not \kbl\ is at most $2^{-\Omega(\sqrt{N}\log N)}$. Therefore, if $m < 2^{cN^{1/5}}$, there is a partition $\varphi\sim{\cal D}$ such that $\Phi_i$ is \kbl\ for every $1\leq i\le m$.  Therefore, by Lemma~\ref{lem:rank}, there is a partition $\varphi\sim{\cal D}$ such that $\rank(M_{g^\varphi}) \le m\cdot s^{O(\log N)}\cdot 2^{N/2 - k}$.  If $m < 2^{c(N^{1/5})}/s^{\log N}$, we have $\rank(M_{g^\varphi}) < 2^{N/2}$, a contradiction to Lemma~\ref{lem:ry}. \qed
\end{proof}

\subsection{Lower Bound against  multilinear $r$-pass   ABPs}
\label{subsec:rpass}
In this section, we extend Theorem~\ref{thm:lb-roabp} to the case of $r$-pass    ABPs. 
Let $P$ be a multilinear $r$-pass    ABP of size $s$ having $\ell$ layers. Let $\pi_1,\pi_2,\ldots,\pi_r$ be the $r$ orders associated with the $r$-pass ABP. Lemmas~\ref{lem:rpasstoformula} and Lemma~\ref{lem:rpass} show that techniques in Section~\ref{subsec:sumroabp} can be adapted to the case of $r$-pass sm-ABPs. Proofs are deferred to the appendix.
\begin{lemma}
\label{lem:rpasstoformula}
Let $P$ be a multilinear $r$-pass ABP of size $s$ having $\ell$ layers computing a polynomial $f\in\mathbb{F}[x_1,\ldots,x_N]$. Then there exists a syntactic multilinear formula 
$\Psi_P=\Psi_1 +\Psi_2 + \cdots + \Psi_t, t= s^{O(r)}$ where each $\Psi_i$ is a syntactic multilinear formula obtained from an ROABP.
\end{lemma}

\makeproof{lem:rpasstoformula}{
Let $P$ be a multilinear $r$-pass ABP of size $s$ computing a polynomial $f\in\mathbb{F}[x_1,\ldots,x_N]$. Then, there exists $i_1,i_2,\ldots,i_{r+1}\in [\ell]$ be such that for $j\in[r]$, the subprogram $[u,v]$ is an   ROABP for any nodes $u$ and $v$ in layers $L_{i_j}$ and $L_{i_{j+1}}$ respectively. The polynomial $f$ computed by $P$ can be expressed as 
\begin{equation}
\label{eqn:k-pass}
f=\sum\limits_{\bar{u}}\prod\limits_{i=1}^r [u_{i},u_{i+1}]
\end{equation}
where the summation is over $\bar{u}=(u_1,u_2,\ldots,u_r)$ where $u_1,u_2,\ldots,u_r$ are nodes in layers $L_{i_1},L_{i_2},\ldots,L_{i_r}$ respectively. As $P$ is a syntactic multilinear ABP, the product term  $\prod_{i=1}^r [u_{i},u_{i+1}]$ in Equation (\ref{eqn:k-pass}) is an ROABP of size atmost $s$ and has a syntactic multinear formula $\Psi_{\bar{u}}$ of size $s^{O(\log N)}$. Thus, $\Psi_P=\Psi_1 +\Psi_2 + \cdots + \Psi_t, t= s^{O(r)}$ where each $\Psi_i$ is a syntactic multilinear formula obtained from an ROABP.  The formula $\Psi$ computing $f$ has size $rs^{O(r+\log N)}$. \qed 
}

\begin{lemma}
\label{lem:rpass}
Let $P$ be a multilinear $r$-pass ABP computing a polynomial $f\in\mathbb{F}[x_1,\ldots,x_N]$ and $\Psi_P=\Psi_1+\Psi_2+\cdots+\Psi_t,~ t=s^{O(r)}$ be the syntactic multilinear formula computing $f$. Let $\varphi\sim{\cal D}$ and $k\leq N^{1/5}$. Then with probability  $1-2^{-\Omega(\sqrt{N}\log N)}$, $\rank(M_{f})\leq |\Psi|\cdot 2^{N/2-k/2}$. 
\end{lemma}
\makeproof{lem:rpass}{
Let $P$ be a multilinear $r$-pass ABP computing a polynomial $f\in\mathbb{F}[x_1,\ldots,x_N]$. By Lemma \ref{lem:rpasstoformula}, $\Psi_P=\Psi_1+\Psi_2+\cdots+\Psi_t,~ t=s^{O(r)}$ be the syntactic multilinear formula computing $f$. Note that $\Psi_i$ is a multilinear formula obtained from an ROABP computing a polynomial $f_i$. By Corollary \ref{cor:rank}, with probability $1-2^{-\Omega(\sqrt{N}\log  N)}$, we have $\rank(M_{f_i^\varphi})\leq |\Psi_i|2^{|X|/2-k/2}$ for $k\leq N^{1/5}$. By sub-additivity in Lemma\ref{lem:sub-aditivity}, $\rank(M_{f^\varphi})\leq \rank(M_{f_1^\varphi})+\ldots+\rank(M_{f_t^\varphi})\leq (|\Psi_1|+|\Psi_2|+\cdots +|\Psi_t|)2^{|X|/2-k/2}\leq |\Psi|\cdot 2^{N/2-k/2}$. \qed
}

Combining the above Lemmas with Lemma~\ref{lem:ry} we get: 
\begin{theorem}
\label{thm:lb-rpass}
Let $f_1,\ldots f_m$ be polynomials computed by multilinear $r$-pass ABPs of size $s_1,s_2,\ldots,s_m$ respectively such that $g= f_1+\cdots + f_m$. Then, $m = \frac{2^{\Omega(N^{1/5})}}{s^{c(r +\log N)}}$, where $c$ is a constant and $s=\max\{s_1,s_2,\ldots,s_m\}$.
\end{theorem}

\makeproof{thm:lb-rpass}{
Suppose, $f_j$ has a multilinear $r$-pass ROABP $P$ of size $s$. Then, by Lemma~\ref{lem:rpasstoformula}, there is a multilinear formula $\Psi_j = \Psi_{j_1} + \dots + \Psi_{j,t}$ computing $f_j$ such that $t \le s^{O(r)}$ and each $\Psi_{j_i}$ is a syntactic  multilinear formula of size $s^{O(\log N)}$ obtained from an ROABP of size at most $s$. By Lemma~\ref{lem:rpass}, $\rank(M_{f_j})\leq t \max_{i}\{|\Psi_{j_i}|\}\cdot  2^{N/2-k/2} \le  s^{O(r+\log N)} 2^{N/2 - k/2}$
with probability atleast $1-2^{-\Omega(\sqrt{N})}$. Therefore, if $s< 2^{o(\sqrt{N})}$,  there is a partition $\varphi\sim{\cal D}$ such that $\rank(M_{g^\varphi}) \le m\cdot s^{O(r+\log N)}\cdot 2^{N/2 - k}$.  If $m < 2^{c(N^{1/5})}/s^{O(r+\log N)}$, we have $\rank(M_{g^\varphi}) < 2^{N/2}$, a contradiction to Lemma~\ref{lem:ry}. \qed
}

\subsection{Lower Bound against sum of $\alpha$-sparse ROABPs}

\label{subsec:sparse}

In this section we prove lower bounds against sum of $\alpha$-sparse ROABPs for $\alpha> 1/10$.
We begin with a version of Lemma~\ref{lem:abptoformula} for sparse ROABPs. 

\begin{lemma}
\label{lem:sparsetoformula}
Let $\alpha \ge 1/10$ and  $P$ be an $\alpha$-sparse ROABP of size $s$ computing a polynomial $f\in\mathbb{F}[x_1,\ldots,x_N]$. Then $f$ can be computed by a syntactic multilinear formula $\Phi$ of size $s^{O(\log d)}$ and depth $O(\log d)$ such that the leaves are labelled with sparse polynomials in $X_i$ for some $1\le i \le d$, where $d = \Theta(N^{\alpha})$.  
\end{lemma}

\makeproof{lem:sparsetoformula}{
The proof is similar to Lemma \ref{lem:abptoformula}. Let $P$ be $\alpha$-sparse ROABP with $d+1$ layers $L_0,L_1,\ldots,L_{d}$ with $s$  and  $t$ as the start and terminal nodes respectively. Let $i=\lfloor{d/2}\rfloor$ and $u_{i_1},u_{i_2},\ldots,u_{i_w} (w\leq s)$ be the nodes at the layer $L_i$. Then,
\begin{equation}
\label{eqn:formula}
f = \sum\limits_{j=1}^{w} [s,u_{i_j}]\cdot[u_{i_j},t]
\end{equation}
where $[u,v]$ is the polynomial computed by the subprogram with start node $u$ and terminal node $v$. By induction on $i$,  Let $\phi_j$ (respectively $\psi_j$) be the formula computing $[s,u_{i_j}]$ (respectively $[u_{i_j},t]$ ). Then $\Phi = \sum_{j=1}^w \phi_j\times \psi_j.$ By induction, it follows  that the resulting formula $\Phi$ has size $s^{O(\log d)}$, depth $O(\log d)$ and is syntactic multilinear. As edge labels in $P$ are sparse polynomials, leaves of $\Phi$ are labeled by sparse polynomials in $\mathbb{F}[X_i]$ for some $ i \in  [d]$. \qed
}

\begin{lemma}
\label{lem:sparse-ub}
Let $P$ be  an $\alpha$-sparse ROABP computing  $f\in\mathbb{F}[x_1,\ldots,x_N]$ and $\Phi$ be the syntactic multilinear formula computing $f$. Let $\varphi\sim{\cal D}$. Then, for any $k \leq N^{(1-\alpha)/4}$, there exists an $i\in [d]$ such that $X_i$ is \kub\ with probability atleast $1- 2^{\Omega(-{N^{1/10}}\log N/16)}$.
\end{lemma}

%
%
%

\makeproof{lem:sparse-ub}{
Let $P$ be an  an $\alpha$-sparse ROABP and $X= X_1 \uplus X_2\uplus\ldots\uplus X_d$. Note that for any $\ell\in[d],~|X_\ell|=\Theta(N^{1-\alpha})$. Let $\varphi\sim {\cal D}$. Let ${\cal{E}}_i$ be the event that the set $X_i$ is not \kub ~ with respect to $\varphi$. For any set $\ell\in[d]$, denote $Y_\ell =\varphi(X_\ell)\cap Y$. Let $\chi =|Y_\ell|$ be a  random variable.  Observe that $\chi$ has the distribution ${\cal{H}}(S,M_1,M_2)$ with\begin{align*}
S &= N-(|X_1|+\cdots+|X_{\ell-1}|)= N-(\ell-1)cN^{1-\alpha} \\
M_1 &= |X_\ell| \\
M_2 &= N/2-(|Y_1|+\cdots+|Y_{\ell-1}|)= N/2 - [(\ell-1)cN^{1-\alpha}]
\end{align*}

For $(\ell-1)<N^\alpha/4c$ where $c$ is the constant hidden in $\Theta$ notation, we have :
\begin{itemize}
\item[(i)] $3N/4\leq S \leq N$; and
\item[(ii)] $0\leq  M_1 = \Theta(N^{1-\alpha})\leq  2N/3 $ when $\alpha\geq 1/10$ for large enough $N$.
\item[(iii)]$S/4\leq N/4\leq M_2 \leq N/2\leq 2S/3 \leq 3S/4 $.
\end{itemize}
By Proposition \ref{prop:hyper}(2), for  $i< {N^\alpha/4c}$, $\alpha> 1/10$ and $k\leq N^{(1-\alpha)/4}$,  $\Pr[{\cal{E}}_i]  \leq  O(N^{-(1-\alpha)/4})$. Let ${\cal{E}}$ be the event that for all $i\in[N^\alpha/4c]$, set $X_i$ is not \kub.
 \begin{align*}
\mathcal{E} &= \mathcal{E}_1 \cap \mathcal{E}_2 \cap \cdots \cap \mathcal{E}_{{N^\alpha}/4c} \\
\Pr[\mathcal{E}] &= \Pr[\mathcal{E}_1 \cap \mathcal{E}_2 \cap \cdots \cap \mathcal{E}_{{N^\alpha}/4c}] \\
&= \Pr[\mathcal{E}_1]\cdot \prod\limits_{i=2}^{{N^\alpha}/4c}\Pr[\mathcal{E}_i\mid \cap_{j=1}^{i-1}\mathcal{E}_{j}] \\
& \leq O(2^{-{N^{1/10}}\log N/16})
\end{align*}
Note $\bar{{\cal{E}}}$ is the event that there exists an $i\in [{N^\alpha}/4c]$ such that $X_i$ is \kub. $\Pr[\bar{{\cal{E}}}] =1- \Pr[{\cal{E}}]\geq 1- 2^{\Omega(-{N^{1/10}}\log N/16)}$.\qed
}

Our first observation is that we can treat the variables sets $X_1,\ldots, X_d$ as blocks $B_1,\ldots, B_d$ as in Section~\ref{subsec:sumroabp}:

\begin{obs}
If $X_r$ is \kub~, then $\Phi$ is \kbl~ for $B=X_r$. 
\end{obs}
Note that for any $t$-sparse polynomial $f$  and any $\varphi\sim {\cal D}$, $\rank(M_{f^\varphi})\leq t$.

\begin{corollary}
\label{cor:rank}
Let $P$ be a $\alpha$-sparse ROABP computing $f$ and $\Phi$ be the multilinear formula obtained from $P$. Let $\varphi\sim{\cal D}$. Then with probability  $1-2^{\Omega(-{N^{1/10}}\log N/16)}$, $\rank(M_{f^\varphi})\leq |\Phi|\cdot t \cdot 2^{N/2-N^{9/40}}$, where $t$ is the sparsity of the polynomials involved in the $\alpha$-sparse ROABP computing $f$.
\end{corollary}
Combining the above with Lemma~\ref{lem:ry}, we get:

\begin{theorem}
\label{thm:sparse-factor}
Let $f_1,\ldots, f_m$ be polynomials computed by $\alpha$-sparse ROABPs of size $s < 2^{N^{9/40}/\log N}$, for $\alpha >1/10$ such that $g = f_1 + \dots + f_m$.
Then  $m \ge 2^{N^{1/11}}$.
\end{theorem}

\makeproof{thm:sparse-factor}{
Let $P_i$ an $\alpha$-sparse factor ROABP computing $f_i$ and $\Phi_i$ is the multilinear formula obtained from $P_i$. Let $t_i$ be the sparsity of $f_i$. By Corollary~\ref{cor:rank}, for $\varphi\sim{\cal D},\rank(M_{f_i^\varphi})\leq |\Phi_i|\cdot t_i \cdot 2^{N/2-N^{9/40}}$ with probability at least $1-2^{-\Omega({N^{1/10}}\log N/16)}$ where $t_i= N^{O(1)}$ is the sparsity of polynomial $f_i$. By sub-additivity
if $m< 2^{N^{1/11}}$, for some $\varphi\sim {\cal D},\rank(M_{g^\varphi}) \le m \cdot \max_i\{|\Phi_i| \}\cdot \max_i\{s_i\}\cdot 2^{N/2 - N^{9/40}}$ with probability $>0$. Therefore $\rank(M_{g^{\varphi}}) < 2^{N/2}$ for some partition $\varphi$, a contradiction to Lemma~\ref{lem:ry}. \qed
}

%
%
%
%
%
%
%
%
%
%
%
%
%
%
%
%

\section{Super polynomial lower bounds for  special a  classes of  multilinear  circuits}

\label{sec:signature}
In this section, we  develop a framework for proving super polynomial lower bound against syntactic multilinear circuits and ABPs based on Raz~\cite{Raz09}.
Our approach involves a more refined analysis of  central paths introduced by Raz~\cite{Raz09}.

\begin{defn}{\em (Central Paths.)}
Let $\Phi$ be a syntactic multilinear formula.  For node $v$ in $\Phi$, let $X_v$ denote the set of variables appearing in the sub-formula rooted at $v$.  A leaf to root path  $\rho = v_1,\ldots, v_\ell$   in $\Phi$ is said to be {\em central}, if $|X_{v_{i+1}}| \le 2 |X_{v_i}|$ for $1\le i\le \ell-1$.  
\end{defn}

For a leaf to root path $\rho: v_1,\ldots, v_\ell$ in $\Phi$, $X_{v_1}\subseteq \ldots \subseteq X_{v_\ell}$ is called the {\em signature} of  the path $\rho$.   A signature  $X_{v_1}\subseteq \ldots \subseteq X_{v_\ell}$ is called central if $|X_{v_{i+1}}| \le 2 |X_{v_i}|$ for $1\le i\le \ell-1$. 
Let $\varphi: X \to Y\cup Z$ be a partition.  A  central signature $X_{v_1}\subseteq \ldots \subseteq X_{v_\ell}$  of a formula $\Phi$ is said to be $k$-unbalanced with respect to $\varphi$ if  for some $i\in[\ell]$, $X_{v_i}$ is $k$-unbalanced with respect to $\varphi$ , i.e., $|\varphi(X_{v_i}) \cap Y - \varphi(X_{v_i}) \cap Z| \ge k $.

The formula $\Phi$ is said to be $k$-weak with respect to $\varphi$, if every central signature that terminates at the root is $k$-unbalanced.  Our first observation is, we can replace central paths in  Lemma~4.1, \cite{Raz09}  with central signatures.  Using the same arguments as in~\cite{Raz09} we get:

\begin{obs}
\label{obs:central-sign-prob}
Let $\varphi:X\rightarrow Y \cup Z$ be a partition of $X=\{x_1,\ldots,x_N\}$. Let $\Phi$ be any multilinear formula compuitng a polynomial $f\in\mathbb{F}[x_1,\ldots,x_N]$. 
\begin{enumerate}
\item If $\Phi$ is $k$-weak with respect to $\varphi$, then ${\rank}(M_{f^\varphi}) \le |\Phi|\cdot 2^{N/2 - k}$.
\item Let $C: X_{v_1} \subseteq X_{v_2} \subseteq \cdots \subseteq X_{v_\ell}$ be a central signature in $\Phi$ such that $ k <|X_{v_1}| \le 2k$.  Then
${\sf Pr}_{\varphi \sim {\cal D}}[\mbox{$C$ is not $k$-unbalanced}] = N^{-\Omega(\log N)}.$
\end{enumerate}
\end{obs}

%

Unfortunately, it can be seen that  even when $P$ is an   ROABP  the number of central signatures in a formula from an ROABP can be $N^{\Omega{\log N}}$. In Section~\ref{subsec:super-poly} we show that a careful bound on the number of central signatures yields super-polynomial lower bounds for sum of ROABPs.

Now, we consider  a subclass of syntactic multilinear circuits   where we can show that the  equivalent formula obtained by duplicating nodes as and when necessary, has small number of  central signatures. 
To start, we consider a refinement of the set of central signatures of a formula, so that Lemma~\ref{lem:rank-full-sign} is applicable to a subset of central signatures in a formula.  

Let $\Phi$ be a syntactically multilinear formula of $O(\log N)$ depth.
Two central paths $\rho_1$ and $\rho_2$ in $\Phi$ are said to {\em meet at $\times$}, if their first common node along leaf to root is labeled by $\times$. A set ${\cal T}$ of central paths in $\Phi$ is said to be {\em \pcover }, if for every central path $\rho\notin {\cal T}$, there is a $\rho' \in {\cal T}$ such that $\rho$ and $\rho'$ meet at $\times$. A {\em \scover} ${\cal C}$ of $\Phi$ is the set of all signatures of the $\pcover$ set $T$ of central paths in $\Phi$.


\begin{lemma}
\label{lem:covering-sign}
Let $\Phi$ be a syntactic multilinear formula. Let $\varphi$ be a partition. If there is a \scover~${\cal C}$ of $\Phi$ such that every signature in ${\cal C}$ is $\kub$ with respect to $\varphi$, then 
$\rank(M_{f^\varphi}) \le |\Phi| \cdot 2^{N/2 - k/2}$.
\end{lemma}

\makeproof{lem:covering-sign}{
We prove by induction on the structure of the formula.  Let $v$ be  the root gate of $\Phi$.  Without loss of generality, assume that $|X_v| > 2k$. Base case is when $X_v$, is $k$-unbalanced. Then clearly,  $\rank(M_{f^\varphi}) \le  2^{N/2 - k/2}$. \\
{\bf Case 1} $v$ is a $\times$ gate with children $v_1$ and $v_2$.
 Then, there is an $i\in\{1,2\}$  such that every central signature containing $X_{v_i}$ is contained in ${\cal C}$. Suppose not, let $\rho_1$ and $\rho_2$ be central signatures in $\Phi$ containing $X_{v_1}$ and $X_{v_2}$ respectively such that $\rho_1,\rho_2\not\in {\cal C}$. Note that $\rho_1$ and $\rho_2$ meet at $\times$ a contradiction to the fact that ${\cal C}$ is an \scover. By induction, we have $\rank(M_{f_{v_i}^{\varphi}}) \le |\Phi_{v_i}|2^{|X_{v_i}|/2 - k/2}$. The required bound follows, since $|X_v| = |X_{v_1}| + |X_{v_2}|$.\\
{\bf Case 2} $v$ is a $+$ gate with children $v_1,\ldots, v_r$.  Then, for every $i \in [r]$,
\begin{itemize}
\item Either every central signature  in $\Phi$ containing $X_{v_i}$ is in ${\cal C}$; or
\item $|X_{v_i}| < |X_{v}|/2$.
\end{itemize}
 In first of the above cases, we have ${\sf rank}(M_{f_{v_i}^{\varphi}}) \le |\Phi_{v_i}|2^{|X_{v_i}|/2 - k/2}$ by inductive hypothesis. In the second case, we have $\rank(M_{f_{v_i}^{\varphi}}) < 2^{|X_v|/4} \le 2^{|X_v|/2 - k/2}$ since $|X_v|> 2k$.
By sub additivity, we have ${\sf rank} (M_{f^\varphi}) \le \sum_{i=1}^r {\sf rank}(M_{f_{v_i}^{\varphi}}) \le \sum_{i=1}^r  |\Phi_{v_i}|2^{|X_v|/2 - k/2} \le |\Phi|\cdot 2^{|X_v|/2 - k/2}$. \qed
}

Let $X_1,\ldots, X_r \subseteq X$, be subsets of variables.  Let $\Delta(X_i, X_j)$ denote the Hamming distance between $X_i$ and $X_j$, i.e, $\Delta(X_i, X_j) = |(X_i \setminus X_j) \cup (X_j \setminus X_i)|$. Let $C_1: X_{11} \subseteq X_{12} \subseteq \cdots \subseteq X_{1\ell}$ and $C_2:  X_{21} \subseteq X_{22} \subseteq \cdots \subseteq X_{2\ell}$ be two central signatures in $\Phi$. Define $\Delta(C_1,C_2) = \max_{1\le i\le \ell} \Delta(X_{1i}, X_{2i})$. Let ${\cal C}$ be \scover ~in $\Phi$.

For $\delta>0$, a {\em $\delta$-cluster} of ${\cal C}$ is a set of signatures $C_1,\ldots, C_t \in {\cal C}$ such that for every $C \in {\cal C}$, there is a $j\in [t]$ with $\Delta(C,C_j) \le \delta$. The following  is immediate:

\begin{obs}
\label{lem:delta-unbalanced}
Let ${\cal C}$ be a \scover, and $C_1,\ldots, C_t$ be a $\delta$-cluster of ${\cal C}$. If $\varphi$ is a partition of $X$ such that for every $i\in [t]$, signature $C_i$ is $\kub$, then  for every $C \in {\cal C}$, signature $C$ is $k-2\delta$ unbalanced. 
\end{obs}  

We are ready to define the special class of  sm-circuits where the above mentioned approach can be applied. 
For $X_1,\ldots, X_r \subseteq X$ and $\delta>0$, a $\delta$-equivalence class of $X_1,\ldots, X_r$, is a minimal set of indices $i_1,\ldots, i_t$ such that for $1\le i\le r$, there is an $i_j, 1\le j\le t$ such that $\Delta(X_i, X_{i_j}) \le \delta$. 

\begin{defn}
\label{def:variable-close} 
Let $\delta \le N \in \mathbb{N}$. Let $\Psi$ be an sm-circuit  with alternating layers of $+$ and $\times$ gates. $\Psi$ is said to be $(c,\delta)$-variable close, if for for every $+$ gate $v = v_{11}\times v_{12} + \cdots + v_{r1}\times v_{r2}$, there  are indices $b_1, b_2, \ldots, b_r \in \{1,2\}$ such that   there is a $\delta$-equivalence class of $X_{v_{1b_1}}, \ldots, X_{v_{rb_r}}$  with at most $c$ different sets. 
\end{defn}

Now, we show that $(c,\delta)$ close circuits have small number of signatures:
\begin{lemma}
\label{lem:delta-signature-bound} Let $\Psi$ be a  $(c,\delta)$-variable close syntactic multilinear arithmetic circuit of size $s$ and depth $O(\log N)$. Let $\Phi$ be the syntactic multilinear formula of size $s^{O(\log N)}$ and depth $O(\log N)$ obtained by duplicating gates in $\Psi$ as and when necessary.  Then there is a \scover~${\cal C}$ for $\Phi$ such that ${\cal C}$ has a $\delta$-cluster consisting of at most $c^{O(\log N)}$ sets.  
\end{lemma}

\makeproof{lem:delta-signature-bound}{
Without loss of generality, assume the root gate of $\Phi$ is a $+$ gate, $\times$ gates have fan-in bounded by $2$, and the layers of $+$ and $\times$ gates is alternating.  We construct the required  $\delta$-cluster $D$  in a top down fashion as follows. 
\begin{enumerate}
\item Initialize $D = X_v$, where $v$ is the root gate in the formula.
\item For a $+$ gate $v = v_{11}\times v_{12} + \dots + v_{r1}\times v_{r2}$, let $b_1,\ldots b_r \in \{1,2\}$  be the indices guaranteed by Definition~\ref{def:variable-close}. Let the $c$ different sets in the $\delta$-equivalence class of $\{X_{v_{1b_1}}, \ldots, X_{v_{rb_r}}\}$ be $X_{i_1,b_{i_1}}, \ldots X_{i_c,b_{i_c}}$. For each partial signature $C'= C_1\subseteq C_{2} \subseteq \cdots \subseteq C_{\ell'}$ such that $C_1 = X_v$ :  Add to set $D$, the signatures $C'^{j} = X_{i_j,b_{i_j}} \subseteq C_1\subseteq C_2 \subseteq \dots \subseteq C_{\ell'}$ for $1\le j\le c$. Now, mark the $+$ gates $v_{i_1b_{i_1}},\ldots, v_{i_c, b_{i_c}}$
\item Repeat $2$ for every marked node until there are nor marked nodes left.  
\end{enumerate}
The set $D$ thus obtained is a $\delta$-cluster for some \scover~${\cal C}$ of $|\Phi|$. $|D|$ is at most $c^{O(\log N)}$, since at every iteration, at most $c$ new signatures might be included for each marked node.  \qed
}

%
%

Finally we conclude with the proof of Theorem~\ref{thm:lb-delta-close}:
\begin{proof}[of Theorem~\ref{thm:lb-delta-close}]
Let $\Psi$ be a $(c,\delta)$ variable close circuit of depth $O(\log N)$. Let $\Phi$ be the formula obtained by duplicating nodes in $\Psi$ as   necessary. By Lemma~\ref{lem:delta-signature-bound},  let $\{C_1,\ldots, C_t\}$ be a $\delta$-cluster of a \scover~ ${\cal C}$ of $\Phi$, for $t = N^{o(\log N)}$. Then, by Observations~\ref{obs:central-sign-prob} and~\ref{lem:delta-unbalanced}, the probability     that there is a signature in ${\cal C}$ that is not $k-2\delta$ unbalanced is at most $t\cdot N^{-\Omega(\log N)} <1$ for $\varphi \sim {\cal D}$. Therefore, there is  a  $\varphi$ such that every signature in $\{C_1,\ldots, C_t\}$ is $k-2\delta$ unbalanced. By Lemma \ref{lem:covering-sign}, there is a  $\varphi$  such that ${\sf rank}(M_{g^\varphi}) \le |\Phi|\cdot 2^{N/2 - (k-2\delta)} \le s^{O(\log N)} 2^{N/2 - k/5} < 2^N/2$ for $s< 2^{k/10\log N}$, a contradiction to Lemma \ref{lem:ry}. \qed
\end{proof}

\bibliographystyle{abbrv}
\bibliography{refbib}

\appendix
\appendixproofsection{Appendix}

\subsection{Oblivious Read-Once Algebraic Branching Programs}
\label{subsec:super-poly}
 In this section we demonstrate the usefulness of central signatures in the case of oblivious ROABPs.  This  exposition is only for demonstrative purpose, the lower bound obtained here is subsumed by Theorem~\ref{thm:lb-roabp}. 
 
Let $P$ be an oblivious ROABP and $\Phi$ the  multilinear formula for $P$ as in Lemma~\ref{lem:abptoformula}. For a gate $v$ in $\Phi$, let $I_v=[i_{v},j_{v}]$ denote the {\em interval associated with gate $v$} as in Section~\ref{sec:explb}.  Let $S_v=\{x_{\ell}\mid i_v\leq \ell \leq j_v\}$ be the set of variables. Note that $X_v \subseteq S_v$. A full central  signature in $\Phi$ is a sequence of sets $S_{v_1} \subseteq S_{v_2} \subseteq \cdots \subseteq S_{v_\ell}$, with $|S_{v_{i+1}}| \le 2|S_{v_i}|$ where $v_1, \ldots, v_{\ell}$ is a leaf to root path in $\Phi$.

\begin{obs}
\label{obs:full-sign}
Let $P$ be an oblivious ROABP computing $f$ and $\Phi$ be a multilinear formula obtained from $\Phi$ for $f$. Let $N$ be a power of $2$ and $C : S_{v_1} \subseteq S_{v_2} \subseteq \cdots \subseteq S_{v_\ell}$ 
be a full central signature  in $\Phi$.  For $i\in \{2,\ldots\ell\}$, we have, either  $|S_{v_i}| = 2|S_{v_{i-1}}|$ or $|S_{v_i}| = |S_{v_{i-1}}|$. Further, the number of full central signatures in $\Phi$ is $O(N)$.
\end{obs}

\makeproof{obs:full-sign}{
Let $P$ be an oblivious ROABP $\Phi$ be a multilinear formula obtained from $\Phi$ for polynomial $f$. Let $N$ be a power of $2$ and $C : S_{v_1} \subseteq S_{v_2} \subseteq \cdots \subseteq S_{v_\ell}$ 
be a full central signature  in $\Phi$. For every $2 \leq i\leq \ell$,
\begin{enumerate}
\item $v_i$ is a $+$ gate :  $v_{i-1}$ is a child of $v_i$. Then, $I_{v_{i-1}}=I_{v_i}$. Since  $S_{v_i}$ is the set of variables corresponding to the interval $I_{v_i}$ and $I_{v_{i-1}}=I_{v_i}$, we have $S_{v_{i-1}}=S_{v_i}$.
\item $v_i$ is a $\times$ gate : $v_{i-1}$ is a child of $v_i$. Let $w$ be the other child of $v_i$. Then, $I_{v_i}=I_{v_{i-1}}\cup I_w$. Since $N$ is a power of $2$, from the construction of $\Phi$ in Lemma \ref{lem:abptoformula}, we have $|I_{v_{i-1}}| = |I_w|$. Hence, $	S_{v_i} = S_{v_{i-1}} \cup  S_w$ implying that $|S_{v_i}| = 2|S_{v_{i-1}}|$. 
\end{enumerate}
For any child $u$ of a $+$ gate $v$, we have $S_u=S_v$. Therefore, we only consider full central signatures where $v_1,v_2,\ldots,v_{\ell}$ are product gates. From construction of $\Phi$ in Lemma \ref{lem:abptoformula}, depth of $\Phi$ is $O(\log N)$ and every $\times$ gate has fan-in 2. Hence, $\ell=O(\log N)$ and number of full central signatures is $2^{O(\log N)}=O(N)$. \qed
}

Let $\varphi :X \to Y \cup Z$ be a partition. We say a gate $v$ in $\Phi$ is $\kw$ with respect to $\varphi$ if every full central signature in $\Phi$ that terminates at $v$ is $\kub$ with reaspect to $\Phi$.

\begin{lemma}
\label{lem:rank-full-sign}
Let $P$ be an oblivious ROABP $\Phi$ be a multilinear formula obtained from $\Phi$ for polynomial $f$. Let $N$ be a power of $2$. If $\varphi :X \to Y \cup Z$ is such that root gate of $\Phi$ is $\kw$ with respect to $\varphi$, then $\rank(M_{f^\varphi}) \le |\Phi|\cdot 2^{N/2 - k/2}$.    
\end{lemma}

\makeproof{lem:rank-full-sign}{
The proof is by induction on the structure of the formula $\Phi$. Let $v$ be the root gate of $\Phi$. Assume that $|S_v| >2k$.
 \noindent \textbf{Case 1 :} $v$ is $\kub$.  Then, ${\sf rank}(M_{f^\varphi}) \le  2^{N/2 - k} \le  |\Phi|\cdot 2^{N/2 - k/2}$. 

\noindent \textbf{Case 2 :} $v$ is a sum gate. Let $v_1,v_2,\ldots,v_r$ be the children of $v$ in $\Phi$, $r\leq w$. Assume w.l.o.g that $v$ is not $\kub$, else apply Case 1. Since $v$ is $\kw$ and gate $v$ is not $\kub$, for every $i\in [r]$ either $v_i$ is  $\kw$ or $k<|S_{v_i}| < |S_{v_{i+1}}|/2$.  In any case, $\rank(M_{f_{v_i}}) \leq |\Phi_{v_i}|\cdot 2^{|S_{v_i}|/2-k/2}$. By sub-additivity, 
\begin{align*}
\rank(M_{f_v}) &\leq \sum_{i=1}^r\rank(M_{f_{v_i}}) \leq |\Phi_{v_1}|2^{|S_{v_1}|/2-k/2} + \cdots + |\Phi_{v_r}|2^{|S_{v_r}|/2-k/2} \\
& \le |\Phi|\cdot 2^{|S_v|/2 - k/2} \le |\Phi|\cdot 2^{N/2 - k/2} \text{~as $|S_v|=N$} 
\end{align*}
\noindent\textbf{Case 3 :} $v$ is a product gate with children $v_1$ and $v_2$. Assume w.l.o.g that $v$ is not $\kub$, else apply Case 1. Since $v$ is $\kw$ and gate $v$ is not $\kub$, atleast one of $v_1$ or $v_2$ is $\kw.$. W.l.o.g, let $v_1$ be $\kw$. By induction, $\rank(M_{f_{v_1}^\varphi}) \leq |\Phi_{v_1}|\cdot 2^{|S_{v_1}|/2-k/2}$ and $\rank(M_{f_{v_2}^\varphi}) \le 2^{|S_{v_2}|/2}$. By sub-multliplicativity, $\rank(M_{f^\varphi})\leq |\Phi_{v_1}|\cdot 2^{|S_{v_1}|/2+|S_{v_2}|/2-k/2} \leq |\Phi|\cdot 2^{|S_v|/2 - k/2} \le |\Phi|\cdot 2^{N/2 - k/2}$ as $S_v = S_{v_1} \cup S_{v_2}$ and $|S_v|=N$. \qed
}

Combining Observation~\ref{obs:full-sign} and Lemma~\ref{lem:rank-full-sign}, we get
\begin{corollary}
Let $f_1,\ldots, f_m$ be  oblivious ROABPs such that $g = f_1+\ldots +f_m$, where $g=g_{1,N}$. Then, $m =  	N^{\Omega(1)}$.
\end{corollary}
\begin{remark}
The above result is only to demonstrate the usefullness of  full central signatures over central paths or central signatures. However, the lower bound above is far inferior to the one in Theorem~\ref{thm:lb-roabp}. 
\end{remark}

\subsection*{Proofs from Section \ref{subsec:sumroabp}}


\appendixproof{lem:abptoformula}
\appendixproof{lem:rank}
\appendixproof{obs:hitting}
\appendixproof{lem:kbl}

\subsection*{Proofs from Section \ref{subsec:rpass}}

\appendixproof{lem:rpasstoformula}
\appendixproof{lem:rpass}
\appendixproof{thm:lb-rpass}

\subsection*{Proofs from Section \ref{subsec:sparse}}

\appendixproof{lem:sparsetoformula}
\appendixproof{lem:sparse-ub}
\appendixproof{thm:sparse-factor}

\subsection*{Proofs from Section \ref{subsec:super-poly}}

\appendixproof{obs:full-sign}
\appendixproof{lem:rank-full-sign}

\subsection*{Proofs from Section \ref{sec:signature}}

\appendixproof{lem:covering-sign}
\appendixproof{lem:delta-signature-bound}

\end{document}